\documentclass[submission,copyright,creativecommons]{eptcs}
\usepackage{breakurl} 
\usepackage[utf8]{inputenc} \usepackage{amssymb, stmaryrd,
  extarrows,amsmath} \usepackage{amsfonts} \usepackage{amsthm}
\usepackage[linesnumbered,ruled,vlined]{algorithm2e}
\usepackage{wasysym} \usepackage{graphics}
\usepackage{graphicx} \usepackage{wrapfig,blindtext}
\usepackage{epstopdf}
\usepackage[x11 names, rgb]{xcolor} \usepackage{tikz}
\usetikzlibrary{decorations,arrows,shapes,automata}
\usepackage{multicol} \usepackage{hyperref}
\usepackage{color}
\usepackage{enumerate}
\usepackage{framed} \usepackage{caption} \usepackage{subcaption}
\usepackage{xcolor} \usepackage{thmtools}
\usepackage{thm-restate}

\newtheorem{theorem}{Theorem}

\newtheorem{definition}[theorem]{Definition}
\newtheorem{example}[theorem]{Example}

\newtheorem{remark}[theorem]{Remark} \newtheorem{cons}{Construction}

\renewcommand{\AA}{\mathcal{A}} \newcommand{\BB}{\mathcal{B}}

\newcommand{\E}{\mathsf{E}} 
 
\newcommand{\NN}{\mathbb{N}} 
 
\newcommand{\ZZ}{\mathbb{Z}} \newcommand{\LL}{\mathcal{L}}
 \newcommand{\sfT}{\mathcal{T}}
 \newcommand{\sfF}{\mathsf{F}}
\newcommand{\sfG}{\mathsf{G}} 
\newcommand{\sfX}{\mathsf{X}} 
\newcommand{\profile}{\sigma} \newcommand{\ebg}{\BB^{c,e}}

\newcommand{\enbuchi}{\mathsf{EnergyBuchi}}

\DeclareMathOperator{\pspace}{\mathsf{PSPACE}}

\DeclareMathOperator{\out}{\mathsf{Payoff}}

\DeclareMathOperator{\NE}{\mathsf{NE}}

\def\set#1{\{#1\}} 
\def\up#1{#1^\mathfrak{U}} \def\ch#1{#1^\mathfrak{C}}

\newcommand{\play}[1]{\langle #1 \rangle}
\newcommand{\cost}{\mathsf{cst}}

\title{Rational Verification in Iterated Electric Boolean
  Games\footnote{The first author is partially supported by the ANR
    project EQINOCS (ANR-11-BS02-004)}}

\author{Youssouf Oualhadj \institute{LACL, U-PEC\\ Paris, France}
  \email{youssouf.oualhadj@lacl.fr} \and Nicolas Troquard
  \institute{LACL, U-PEC\\ Paris, France}
  \email{nicolas.troquard@lacl.fr} } 
\def\titlerunning{Rational Verification in Iterated Electric Boolean
  Games} 
\def\authorrunning{Y. Oualhadj \& N. Troquard}
\begin{document}

\maketitle

\begin{abstract}
  Electric boolean games are compact representations of games where the
  players have qualitative objectives described by LTL formulae and have limited
  resources. We study the complexity of several decision problems
  related to the analysis of rationality in electric boolean games
  with LTL objectives. In particular, we report that the problem of
  deciding whether a profile is a Nash equilibrium in an iterated
  electric boolean game is no harder than in iterated boolean games
  without resource bounds.  We show that it is a $\pspace$-complete
  problem. As a corollary, we obtain that both rational elimination
  and rational construction of Nash equilibria by a supervising
  authority are $\pspace$-complete problems.
\end{abstract}

\section{Introduction}
We study multiagent systems populated with self-interested agents who
interact repeatedly and are limited in their actions by a limited
amount of energy. We investigate the computational aspects of deciding
whether a collective, non-cooperative, behaviour is rational.



\paragraph{Electric boolean games} The formalism under consideration
was introduced in the second part of \cite{HarrensteinTW15} but the
decision problems were left open. They extend naturally the models of
multi-player boolean games \cite{DBLP:conf/ecai/BonzonLLZ06}, one-shot
electric games \cite{HarrensteinTW15}, and iterated boolean
games~\cite{GutierrezHW15}.  Boolean games have occupied an important
position in the recent formal AI literature.  This line of work is an
effort in formalisation of game theoretical situations with boolean
games (see previously cited work and
e.g.,~\cite{Wooldridge2013418,DBLP:conf/atal/DunneHKW08}).



Strategically, the players in Iterated Electric Boolean Games
(Sec.~\ref{sec:ebg}) are intricately mixing qualitative and
quantitative considerations. Not only do they need to find a strategy
that helps them satisfy their qualitative objective over time, they
need to do so, seeking to keep the interaction alive so as not to run
out of energy and fail to be able to perform a single action. This can
be illustrated by the next simple example.

\begin{example}\label{ex:spoiled-kids}
  Isabella and Jules are two demanding kids. Isabella's objective
  towards happiness is to be granted a new comic book on a regular
  basis, and Jules' objective is to be granted a new jigsaw puzzle
  just as often. Their mom's objective is naturally to have all
  requests eventually fulfilled. Whether they ask for a new item or
  not, it costs zero to the kids either way. They never incur any
  costs. Buying a new comic book however, will cost $\$4$ to their
  mother, and getting a new jigsaw puzzle will cost her $\$6$. Each
  day, each item that is not bought will earn Mom $\$1$. Isabella and
  Jules, being what they are, decide that their behaviour to satisfy
  their objective is to ask a new item all the time. Fortunately, Mom
  is going to cope with it by waiting $5$ days, buying a new comic
  book and a new jigsaw puzzle on the $6$-th day, and repeating. It
  results in a collective behaviour which is rational as we shall
  explain later on.
\end{example}

\paragraph{Boolean games as compact game representations}
Solving problems on an input only makes sense when the input is
reasonable.  Possible worlds and relational semantics are commonly
used to model multiagent systems. However, describing a complex system
in terms of possible worlds is often unpractical. In fact, the size of
the description of a system as a transition system typically grows
exponentially in the number of variables in the system.  For instance,
model checkers for Alternating-time Temporal Logic make use of
Reactive Modules~\cite{DBLP:conf/cav/AlurHMQRT98} or Interpreted
Systems~\cite{DBLP:conf/cav/LomuscioQR09} to overcome the
difficulty. The powers of agents and coalitions are derived from the
ability to control the value of some variables, thus bringing about
some change to the system.
Boolean games~\cite{harrenstein04phd, DBLP:journals/aamas/BonzonLLZ09}
are such compact representations which in addition also integrate
agents' preferences. They recently have been widely used to study
various phenomena relevant to artificial
intelligence~\cite{DBLP:conf/atal/DunneHKW08,DBLP:journals/ijar/BonzonLL09,
  DBLP:journals/ijait/BonzonDL10, DBLP:conf/ijcai/GrantKWZ11,
  Wooldridge2013418}.

Boolean Games are multi-player games where each player controls a set
of propositional variables and has a qualitative preference
represented by a propositional formula over the set of variables in
the system. An action for a player is to assign a valuation to the
propositional variables she controls. Iterated Boolean
Games~\cite{GutierrezHW15} are a variant of Boolean games where the
players repeat the interaction infinitely often, and where their
qualitative objectives are represented as LTL formulas over the set of
variables in the system.

Electric Boolean Games~\cite{HarrensteinTW15} are an extension of
Boolean Games where agents are assigned an initial energy endowment
and taking actions has a cost, positive or negative.  Already
in~\cite{HarrensteinTW15}, the authors define an iterated version of
Electric Boolean Games, but they do not investigate their strategic
aspects.

\paragraph{Design of safe computer systems}
In theoretical computer science, and particularly in the design and
verification of computer systems, two-player zero-sum games have been
extensively studied and used with great
success~\cite{AMP94,TA99}. Recently, researchers have brought their
attention to introducing quantitative restrictions for the
players. For instance games where the system has to accomplish a task
while maintaining its resource level above zero was modelled using Mean
payoff Parity games~\cite{CHJ05}, or Energy Parity
Games~\cite{CD12}. This line of work was naturally extended by the
study of the so-called multi-objective games with actual
implementation~\cite{BrazdilCFK15}. In a multi objective game, a
protagonist player wants to achieve a conjunction of goals, and the
antagonist player wants to achieve the exact opposite. Nevertheless,
the pessimistic assumption that a system and its environment always
have opposite interests is not always realistic. Therefore,
multiplayer games seem to be a more suitable
formalism~\cite{BrenguierCHPRRS16}. Indeed, the environment is
considered to be another player with her own goal. In order to study
those games, the solution concept of choice was Nash equilibria as it
is a sensible formalisation of rationality~\cite{BrihayePS13}.
In an electric boolean game,
each agent has to partake in a cooperation that keeps the system
alive. Namely, every single player has to make sure that none of the
other players is running out of resource. This approach can be seen as
an intermediate setting between non-cooperative and cooperative
games. Actually, this can also be seen as a new definition of
multi-objective games in the setting of multi-player games; Every player
has a personal goal with no incentive to cooperate and second goal
where it is best for her to cooperate.

\paragraph{Engineering multiagent systems}
Some plays of a game may appear better than others by some supervising
authority. Some strategic equilibria in a game may be undesirable,
while play which are not equilibria might be seen as desirable. A
supervising authority could have the power to redistribute the
resources available in the system so as to achieve better equilibria
from their point of view. Dealing with resources such as energy, it
then becomes interesting to study how much different the game would
be, were the endowments of the players be different. As
in~\cite{HarrensteinTW15}, it is very natural to consider resource
redistributions that allow one to eliminate `bad' equilibria and/or
construct `good' equilibria.

Apart from~\cite{Wooldridge2013418} and~\cite{HarrensteinTW15},
looking into ways of engineering a game's outcome has also been
considered in \cite{DBLP:conf/concur/AlmagorAK15}. The authors propose
a framework where the winning conditions can be modified at a cost,
thus changing the strategic equilibria of the game.

\paragraph{Contributions} Our main result is the $\pspace$ membership
for rational verification i.e., given a strategy profile decide
whether it is a Nash equilibrium (Sec.~\ref{sec:ne}). Note that the
computational complexity in the electric case matches the one in the
non-electric case.  Our proof differs from the one
in~\cite{GutierrezHW15} for the non-electric case. Indeed, a
straightforward adaptation of their proof would fail for it relies on
a translation of the input into a well chosen LTL formula.
In the electric case, one has to pay particular attention to the
electric constraints (c.f., Ex.~\ref{ex:dev}).
This is a quantitative ingredient that is absent from LTL.
We overcome this difficulty as follows.
We construct a one-player game played on a weighted graph. This allows
us to encode the behaviour of the possible deviator together with the
electric constraints in an existing formalism, viz., \emph{Energy
  Büchi games}~\cite{ChatterjeeRR12}.  We prove that a rational
deviation exists iff this one-player game contains a winning strategy.
The size of the constructed one-player game may be exponential in the
size of the input. However, on-the-fly automata-theoretic techniques
allow one to maintain a $\pspace$ upper-bound for the problem of
finding a winning strategy. Finally, to decide in $\pspace$ whether a
strategy profile is a Nash equilibrium, it suffices to guess a
deviator and check whether she has a winning strategy in her
one-player game.


Solving rational verification facilitates the access to more problems.
We show (Sec.~\ref{sec:engineering}) that the problems of resource
redistribution come out as corollaries. We leave open the more
challenging problem of rational synthesis for which rational
verification is a stepping stone; Rational verification is to model
checking what rational synthesis is to model synthesis.

A full version is  available in~\cite{DBLP:OT16}.


\section{Iterated Electric Boolean Games}
\label{sec:ebg}

\begin{definition}[Electric Boolean Games]
  An \emph{electric boolean game} (EBG for short) is a tuple
  $\BB = (N, A, \Phi, c, e )$ where:
  $N = \{1,\cdots,n\}$ is a finite set of players.
  $A=\cup_{i=1}^n A_i$ with $A_i$ are the
  atoms controlled by player $i$ and $(A_1,\cdots,A_n)$ forms a partition of $A$.
  $\Phi = \set{\phi_1,\cdots,\phi_n}$ where $\phi_i$ is the objective of player $i$.
  $c : A \times \{\bot,\top\} \to \ZZ$ is a cost function.
  $e: N \to \NN$ is an endowment function.
\end{definition}

We denote $\sfT$ the set $\set{\bot,\top}$ and for any set $E$,
$\sfT^{E}$ the set of mappings from $E$ to $\sfT$,
the set of all the finite sequences over $E$ is $E^*$ , and
$E^\omega$ is the set of all the infinite sequences over $E$.

Let $X$ be a set of atomic propositions, a valuation of $X$ is a
total function $v\in \sfT^{X}$.
The cost of a valuation $v$
is given by $\cost(v) = \sum_{p \in X} c(p, v(p))$. An \emph{action} of player $i$
is to assign a valuation to each variable in the set $A_i$ of the
atoms she controls. 

We consider the setting of concurrent and infinitely repeated electric
boolean games, where players choose their actions simultaneously and
for an infinite duration. 
We consider objectives in $\Phi$ which are specified by LTL formulas
over the atoms of $A$ (\cite[Chap.~5]{Baier:2008:PMC:1373322}).
Formulas of LTL are defined by the following grammar: $\phi ::= p \mid
\phi \land \phi \mid \lnot \phi \mid \mathsf{X}\phi \mid \phi
\mathsf{U}\phi$ where $p \in A$. The other propositional operands and
temporal operators ($\sfF$, $\sfG$) can be defined as usual.

We need to introduce some useful terminology to talk about repeated
games and define the semantics of LTL formulas over
$\left(\sfT^A\right)^\omega$.

A \emph{history} in a repeated electric boolean game is a word in
$\left(\sfT^{A}\right)^*$. That is, a finite sequence of
valuations for the set $A$ of boolean variables.  A \emph{play}
is an infinite sequence in $\left(\sfT^{A}\right)^\omega$.
Given a play $\rho$, we note $\rho[t]$ the $t$-th valuation
function in $\rho$. We note $\rho[t\ldots]$ the suffix of $\rho$
starting at $\rho[t]$, and $\rho[\ldots t]$ the prefix of $\rho$
ending at $\rho[t]$ which is a history of size $t+1$.

LTL objectives are evaluated over a play $\rho$ of the game. For $p
\in A$, and for $\phi$ and $\psi$ two LTL formulas:
\begin{align*}
  \rho &\models p \text{ iff } \rho[0](p) = \top&
  \rho &\models \lnot \phi \text{ iff } \rho \not\models \phi&\\
  \rho &\models \mathsf{X} \phi \text{ iff } \rho[1\ldots] \models \phi&
  \rho &\models \phi\land\psi  \text{ iff } \rho \models \phi \text{
    and } \rho \models \psi\\
  \rho &\models \phi \mathsf{U} \psi \text{ iff } \exists i \geq 0,~
         \rho[i\ldots] \models \psi \text{ and } \forall 0 \leq j < i,~
         \rho[j\ldots] \models \phi&  
\end{align*}

The formula $\mathsf{X}\phi$ holds true on $\rho$ if $\phi$ is
true next. The formula $\phi \mathsf{U}\psi$ holds true on $\rho$ if
$\phi$ is true at least until $\psi$ is true.

In order to play, the players choose their actions according to a
strategy. A \emph{strategy} for player $i$ is a mapping that takes as
input a history and outputs a valuation for each atom controlled by
player~$i$. Formally a strategy $\sigma_i$ for player $i$ is a mapping
$\sigma_i : \left(\sfT^{A}\right)^* \to
\sfT^{A_i}$.
We note $\Sigma_i$ the set of strategies of player $i$.

A \emph{strategy profile} ${\sigma}$ is a vector
$(\sigma_1, \cdots, \sigma_n)$ specifying one strategy $\sigma_i$ for
each player $i \in N$. Given a strategy profile
${\sigma} = (\sigma_1, \cdots, \sigma_n)$ and a strategy $\tau_i$ for
player $i$, we note $(\tau_i, \sigma_{-i})$ the strategy profile
$(\sigma_1, \cdots, \tau_i, \cdots, \sigma_n)$.
Each strategy profile induces a
play, and since we consider pure strategies, there is
one and only one such play consistent with $\sigma$. We denote
$\play{{\sigma}}$ the play induced by the profile ${\sigma}$. 
It is defined inductively as follows: if $p \in A_i$ then
$\play{{\sigma}}[0](p) = \sigma_i(\epsilon)(p)$, and for $t\ge 0$,
$\play{{\sigma}}[t+1](p) = \sigma_i(\play{{\sigma}}[\ldots t])(p)$.

The endowment $e(i)$ of each player $i$
specified in the definition of an electric boolean game, represents the
initial resources of the player. While playing the game following a
strategy, this endowment grows as the player takes an action of
negative cost and shrinks as the player takes an action of positive
cost.

We will say that the strategy profile $\sigma$ is feasible in an
iterated EBG if it does not over-consume the endowed resources, in the
sense that, every player's strategy $\sigma_i$ 
can be infinitely
executed without ever causing the player's compound endowment to go
under $0$.
We make it more formal.

Consider an EBG $(N, A, \Phi, c, e )$ and a strategy profile
$\sigma$. The \emph{compound endowment} of player $i$ at the $t$-th step of
the play $\play{{\sigma}}$ is defined with $\E^\sigma_i(0) = e(i)$,
and
\[\E^\sigma_i(t+1) = \E^\sigma_i(t) - 
\cost(\sigma_i(\play{{\sigma}}[\ldots t]))\]
Thus, the strategy profile ${\sigma}$ is \emph{feasible} iff for
each player $i \in N$, 
and for all $t \ge 0$ we have
$ \E^\sigma_i(t) \geq 0$.
In the strategy profile $\sigma$, we say that $\tau_i$ is a
\emph{feasible deviation} for player~$i$ iff $(\tau_i, \sigma_{-i})$
is a feasible strategy profile.

Once an objective $\phi_i$ and a strategy profile ${\sigma}$ are
fixed, the payoff of ${\sigma}$ for player $i$ is defined as
follows:
\[
\out_i({\sigma}) =
\begin{cases}
  1 & \text{ if } \sigma \text{ is feasible, and } \play{\sigma} \models \phi_i \enspace,\\
  0 & \text{ otherwise.}
\end{cases}
\]
In the strategy profile $\sigma$, we say that $\tau_i$ is a
\emph{rational deviation} for player~$i$ iff
$\out_i((\tau_i,\sigma_{-i})) > \out_i(\profile)$.

\begin{example}\label{ex:spoiled-kids-form}
We formalise the game of Example~\ref{ex:spoiled-kids} and model a
strategy for the three participants.  Let $\ebg$ be an EBG $(N, A,
\Phi, c, e)$ where $N = \{I,J,M\}$, $A_I = \{r_I\}$, $A_I = \{r_J\}$,
$A_M = \{g_I, g_J\}$. Evaluated to $\top$, the atoms $r_I$, $r_J$,
$g_I$, $g_J$, respectively represent the facts that Isabella asks for
a comic book, Jules asks for a jigsaw puzzle, Mom buys a comic book,
and Mom buys a jigsaw puzzle. The costs are given by $c(r_I,\top) =
c(r_I,\bot) = c(r_J,\top) = c(r_J,\bot) = 0$, and $c(g_I, \bot) =
c(g_J, \bot) = -1$, $c(g_I, \top) = 4$, and $c(g_J, \top) = 6$. We
suppose that $e(I) = e(J) = e(M) = 0$.  The objectives are given as
$\Phi_M = \sfG((r_I \to \sfF(g_I)) \land (r_J \to \mathsf{F}(g_J)))$,
$\Phi_I = \sfG\sfF(g_I)$, and $\Phi_J = \sfG\sfF(g_J)$. The strategies
of the kids continuously asking a new item and of the Mom buying one
comic book and one jigsaw puzzle every $6$ days result in a strategy
profile whose payoff is $1$ for everyone.
\end{example}


\begin{wrapfigure}{r}{0.35\textwidth} 
\vspace{-1cm}
  \begin{center}
  \begin{subfigure}[b]{0.2\textwidth}
    \begin{tikzpicture}[>=stealth', scale = .75]
      \node (0) at (0,0) [draw, state, initial, initial text=] {$0$};
      \node [above] at (0.north) {$ r_I$};
      \draw[->] (0)to [out = -30, in = 30, loop] node [] {} (0);
    \end{tikzpicture}
    \caption{Isabella's strategy.}
  \end{subfigure}
  \begin{subfigure}[c]{0.2\textwidth}
    \begin{tikzpicture}[>=stealth', scale = .75]
      \node (0) at (0,0) [draw, state, initial, initial text=] {$0$};
      \node [above] at (0.north) {$r_J$};
      \draw[->] (0)to [out = -30, in = 30, loop] node [] {} (0);
    \end{tikzpicture}
    \caption{Jules' strategy.}
  \end{subfigure}
\end{center}
\hfill
\begin{center}
  \begin{subfigure}[c]{0.4\textwidth}
    \begin{tikzpicture}[>=stealth', scale = .85]
      \node (0) at (0,0) [draw, state, initial, initial text=] {$0$};
      \node [above] at (0.north) {$(\lnot g_I,\lnot g_J)$}; \node (1) at
      (2,0) [draw, state] {$1$}; \node [above]at (1.north)
      {$(\lnot g_I,\lnot g_J)$}; \node (2) at (4,0) [draw, state] {$2$};
      \node [above] at (2.north) {$(\lnot g_I,\lnot g_J)$}; \node (3) at
      (4,-2) [draw,state] {$3$}; \node [below] at (3.south)
      {$(\lnot g_I,\lnot g_J)$}; \node (4) at (2,-2) [draw,state] {$4$};
      \node [below] at (4.south) {$(\lnot g_I,\lnot g_J)$}; \node (5) at
      (0,-2) [draw,state] {$5$}; \node [below] at (5.south)
      {$(g_I,g_J)$};
      \draw[->] (0) to [] node [] {} (1); \draw[->] (1) to [] node []
      {}(2); \draw[->] (2) to [] node [] {} (3); \draw[->] (3) to []
      node [] {} (4); \draw[->] (4) to [] node [] {} (5); \draw[->]
      (5) to [] node [] {} (0);
    \end{tikzpicture}
    \caption{Mom's strategy.}
    \label{fig:strat-mom}
    \vspace{-.5cm}
  \end{subfigure}%
\end{center}
\caption{A finite memory profile seen as finite
  graphs.\label{fig:straMach}}
\vspace{-1cm}
\end{wrapfigure}
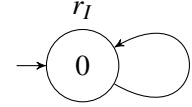
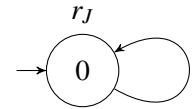
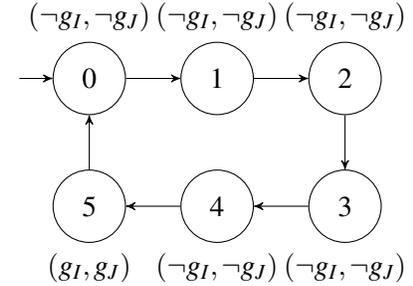

The strategies suggested at the end of
Example~\ref{ex:spoiled-kids-form} are depicted in
Figure~\ref{fig:straMach}. They are instances of what we call
\emph{finite memory} strategies.
We formalise the class of finite memory strategies next.

\begin{definition}[Finite memory strategy]
  \label{def:finite-memory-strategy}
  Let $i\in N$ be a player,
  a finite memory strategy $\sigma_i$ for player $i$ consists of a
  finite set $M$ called the memory, an initial memory state $m^{in}$
  in $M$, a mapping $\up{\sigma_i} : M\times \sfT^{A} \to M$ called the
  update function, and a mapping $\ch{\sigma_i} : M \to \sfT^{A}$ called
  the choice function.
\end{definition}

We say that $(\sigma_1,\cdots,\sigma_n)$ is a \emph{finite memory} profile if for every
$i\in N$, $\sigma_i$ is a  finite memory strategy.
For instance, in the strategy of Figure~\ref{fig:strat-mom},
the set $M$ is $\set{0,1,2,3,4,5}$, the initial memory state is $0$, the
update function is the edge relation and the choice function is
illustrated by labels next to vertices\footnote{We omit the labels on
  the edges to highlight that  for each player the update function
  depends only on the current memory state.}.


\section{Nash Equilibria in Electric Boolean Games}
\label{sec:ne}

In~\cite{HarrensteinTW15}, the authors introduced iterated electric boolean
games but did not study their strategic aspects. Hence
no solution concept was defined. However, the concept of Nash
equilibria is one of most natural concept in multiplayer games.


\begin{definition}[Nash equilibrium]\label{def:ne}
  Let $\ebg$ be an EBG and $\profile$ be a strategy profile. We say
  that $\profile$ is a \emph{Nash equilibrium} iff the following holds:
  \begin{enumerate}
  \item $\forall t\ge0,~\forall i\in N,~\E_i^\sigma(t) \ge 0$,
  \item
    $ \forall i\in N,~\forall \tau_i\in\Sigma_i~,
    \out_i((\tau_i,\sigma_{-i})) \le \out_i(\profile) $.
  \end{enumerate}
\end{definition}

Using our terminology, $\profile$ is a Nash equilibrium in $\ebg$ if
and only if it is feasible and there is no rational deviation for any
player. We note $\NE(\ebg)$ the set of Nash equilibria in the game
$\ebg$. For instance, the strategy profile depicted in
Figure~\ref{fig:straMach} is a Nash equilibrium in the game of
Examples~\ref{ex:spoiled-kids} and~\ref{ex:spoiled-kids-form}

\begin{definition}[Nash Equilibrium Membership]
  Let $\ebg$ be an electric boolean game, and $\profile$ be a finite
  memory strategy profile. The Nash Equilibrium Membership (NEM)
  problem asks whether $\profile \in \NE(\ebg)$.
\end{definition}

In order to build intuition regarding deviations, consider the
following example

\begin{wrapfigure}{r}{0.4\textwidth} 
\vspace{-1cm}
\begin{center}
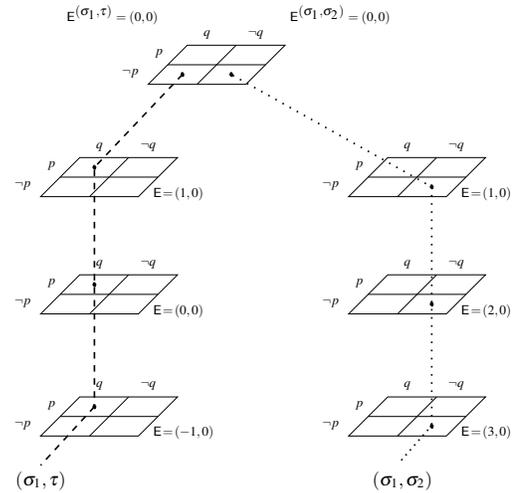
\end{center}
\vspace{-.5cm}
\caption{plays induced by the profiles $(\sigma_1,\sigma_2)$ and
  $(\sigma_1,\tau)$.} 
\label{fig:plays}
\vspace{-1.2cm}
\end{wrapfigure}


\begin{example}
  \label{ex:dev}
  Let $\ebg$ be the following two-player game, 
  \begin{align*}
    &A_1 = \set{p},~ A_2=\set{q}\enspace,\\
    &\phi_1 \equiv \sfG\left((q\to \sfX p) \land (\lnot q \to
      \sfX\lnot p)\right),~\phi_2 \equiv \sfG q\enspace,\\
    &c(p,\top) = 1 ,~c(p,\bot) = -1,~c(q,\top) = c(q,\bot) =
      0,~e(1) = e(2) = 0\enspace.
  \end{align*}

Consider the following strategy
$\sigma_1$ for player 1 that assigns
  $\top$ to $p$ iff $\top$ was assigned to $q$ the previous
  round. We also consider the strategy $\sigma_2$ for player 2 that
  always assigns $\bot$ to $q$.

  We argue that the profile $(\sigma_1,\sigma_2)$ is a Nash
  equilibrium. Clearly $(\sigma_1,\sigma_2)$ is feasible. Let us show
  that player 2 does not have a rational deviation. In order to
  increase her payoff, player 2 has to always assign $\top$ to $q$,
  call this new strategy $\tau$. However, the deviation $\tau$ is not
  feasible.  Indeed, player 1 is still following $\sigma_1$, we obtain
  \begin{align*}
    &\sigma_1
    \left(\epsilon
    \right)(p) = \bot    
     \text{ with }\E_1^{(\sigma_1,\tau)}(1) = 1 \enspace,\\
    &\sigma_1
    \left(
      \set{(p,\bot), (q,\top)}
      \right)(p) = \top
    \text{ with } \E_1^{(\sigma_1,\tau)}(2) = 0 \enspace,\\
    &\sigma_1\left(
      \set{(p,\bot), (q,\top)}\set{(p,\top),(q,\top)}
    \right)(p) = \top        
    \text{ with } \E_1^{(\sigma_1,\tau)}(3) = -1 \enspace,
  \end{align*}
  showing that the compound endowment drops below 0 after the third
  round. The plays induced by the two profiles are depicted in
  Figure~\ref{fig:plays}.
\end{example}

This example shows that in order to perform a rational deviation, a
player has to check the endowment of all the players and not only her
own. We are now ready to state the main theorem of this paper.

\begin{restatable}{theorem}{thmMain}
\label{th:NEM-PSPACEc}
  NEM is a $\pspace$-complete problem. It is $\pspace$-hard even
  when there is only one player.
\end{restatable}

To prove the theorem, we exhibit two
constructions, c.f.~Construction~\ref{cons:feas}, and
Construction~\ref{cons:dev}. The former allows one to check the feasibility of
a profile, while the latter allows one to check the existence of a
rational deviation.

In Section~\ref{sec:check-feas-pspace}, and
Section~\ref{sec:check-dev-pspace}  we 
let $\ebg$ be an EBG, and
$\sigma$ be a finite memory 
profile.  Let also $(M_i, m_{i}^{in},\up{\sigma}_i,\ch{\sigma}_i)$ be the finite
memory strategy of player $i$ in the profile $\sigma$.

\subsection{Checking feasibility in PSPACE}
\label{sec:check-feas-pspace}

We say that $G$ is a \emph{$d$-weighted graph} if $G$ is associated with 
a weight function $w:E\to\ZZ^d$. For a vertex $u$ and a vector $w_0$ in
$\NN^d$, 
a subset $C$ of $V$ is a nonnegative
reachable cycle from $u$
if the following holds. 
(i)~There exists $v$ in $C = \set{u_j \mid l \le j \le k}$, and a
path $u_0,\cdots,u_l,\cdots u_{k}$ such that $u_0 = u$, $u_l = v$,
and $u_{k} = v$. 
(ii)~For all $0\le t\le k-1$ we have $w_0 - \sum_{j=0}^t w(u_j,u_{j+1})
\ge\set{0}^d$, and $\sum_{j=l}^{k-1} w(u_j,u_{j+1})\le\set{0}^d$.
Positive cycles are defined as expected.

In order to prove Proposition~\ref{prop:feasCheck} we use the results
of~\cite{KosarajuS88}.  In particular, given a $d$-weighted graph $G$,
we can detect a \emph{nonnegative reachable cycle} in polynomial time
in the size of $G$.\footnote{The result of \cite{KosarajuS88} is to
  find $0$-cycles. To find nonnegative cycles, it suffices to
  transform a weighted graph $G$ into $G'$ by adding a reflexive edge
  of weight $-1$ to every vertice. This is a polynomial
  transformation. $G$ has a nonnegative cycle iff $G'$ has a
  zero-cycle.}


Our approach consists in constructing 
a $n$-weighted graph $G[\sigma]$ from the
finite memory profile $\sigma$. This is achieved by
Construction~\ref{cons:feas}. We show that $G[\sigma]$ contains such a
cycle iff $\sigma$ is feasible.

We start first by giving the details of how $G[\sigma]$ is obtained.

\begin{cons}
\label{cons:feas}
$G[\sigma]$ consists of a finite set of vertices $V$, an edge relation
$E \subseteq V\times V$, and weight function $w : E \to \ZZ^{n}$.
$G[\sigma]$ is obtained as follows:
\begin{itemize}
\item The vertices are $V = \prod_{i\in N}M_i$.
\item For $v\in V$ we denote $v_i$ the $i$-th component of $v$. Let
  $(u,v) \in V\times V$ be a couple of vertices, $(u,v)$ is an edge in
  $E$ if for each $i\in N$ we have $\up{\sigma}_i(u_i, X) = v_i$ where
  $X = \bigcup_{j \in N}\ch{\sigma}_j(u_j)$ is the complete valuation
  over $A$ prescribed by the profile $\sigma$.
\item Finally, for $(u,v) \in E$,
  \[
  w(u,v) = \left(
    \cost(\ch{\sigma}_1(u_1)),\cdots,\cost(\ch{\sigma}_{n}(u_{n})) \right)
  \enspace.
  \]
\end{itemize}
\end{cons}


The following lemma states the key property of
Construction~\ref{cons:feas}.

\begin{restatable}{lemma}{lmConsOne}
  \label{prop:feasCor}
  The finite memory strategy profile $\profile$ is feasible iff
  $G[\sigma]$ has a nonnegative reachable cycle from 
  $u^0 = (m^{in}_1, \ldots, m^{in}_n)$ with initial credit $e$.
\end{restatable}

A consequence of the above lemma is

\begin{restatable}{proposition}{PropConsOne}
  \label{prop:feasCheck}
  We can check in $\pspace$ whether $\sigma$ is feasible.
\end{restatable}

\subsection{Checking the existence of rational deviation in PSPACE}
\label{sec:check-dev-pspace}
Now that we can check whether a profile is feasible, we need to show
how to check the existence of rational deviation for a player.

We recall  that $\ebg$, $\sigma$, and $\sigma_i =
(M_i, m_i^{in},\up{\sigma}_i,\ch{\sigma}_i)$ are still fixed.

We need to introduce some technical material.
A B\" uchi automaton $\AA$ is a tuple $\AA = (Q,q_0, A, \Delta, F)$ where
the $Q$ is a finite set of states, $q_0$ is an initial state, $A$ is a
finite alphabet, $\Delta$ is relation in $Q\times A\times Q$, and $F$
is a subset of states called accepting. We say that an infinite word $w$
is recognised by $\AA$ if there exists an infinite path $\rho$ in $\AA$
labelled by $w$ such that  $\rho$ visits states in $F$ infinitely many
times. We also say that $\rho$ is a run induced by $w$ on $\AA$.
We define $\LL_\AA$ as the set of words recognised by $\AA$.
The reason we need Büchi automata is their strong link with
LTL. Indeed, any LTL formula $\phi$, can be
associated to a Büchi automaton accepting all its models. The
following theorem formalises this idea.
\begin{theorem}
  \label{th:ltlaut}
  Let $\phi$ be a LTL formula, there exists a Büchi automaton
  $\AA_{\phi}$ accepting the language $\LL_{\phi}$ consisting of all
  the models of $\phi$.
\end{theorem}

The other formalism is \emph{one-player games}.
Let $G=(V,E, W)$ be a graph with a set of vertices $V$,
a set of edges $E\subseteq V\times V$,
and winning objective $W \subseteq V^\omega$.  Strategies for these
games are formalised by the following mapping $V^*V \to V$.
Let $\sigma$ be a strategy for the player, and $u_0$ a vertex in $V$.
The play $\rho$ starting in $u_0$ and consistent with $\sigma$ is
obtained as follows:
$\rho[0] = u_0$, and for all $i>0$, $\sigma(\rho[\ldots i])$.
The player wins if the play $\rho$ is in $W$.
A strategy $\sigma$ is winning for the player from $u_0$
if the play consistent with
$\sigma$ is in $W$. Finite memory strategies can be
defined in a similar fashion as for EBGs.
In this paper, we use the so-called multi-objective games. Those are
games where the player has to fulfil a combination of objectives at once. 

 {\bf Büchi objectives.} We choose a set $F \subseteq V$ of accepting
 vertices. The winning objective $W$ is $(V^*F)^\omega$. We denote
 this winning objective $\mathsf{Buchi}$.

{\bf Energy objectives.} Let $d>0$ be a natural, $w_0 \in \NN^d$ be an
initial vector, and $w : E \to \ZZ^{d}$ be an energy function. The
winning objective is the set
$\set{u_0u_1\cdots \in V^\omega\mid \forall k\ge i,~w_0 - \sum_{i=0}^k
  w(u_i,u_{i+1}) \ge \set{0}^d}$.
We denote this winning objective $\mathsf{Energy}$.

The winning objective we are interested in is $\enbuchi$ defined by
$\mathsf{Buchi} \cap\mathsf{Energy}$.

Roughly speaking, given a profile $\sigma$ and a player $i$, we
construct a one-player $\enbuchi$ game $G[\sigma_{-i}]$. The purpose of this
game is to contain a winning strategy iff a rational
deviation exists. Moreover, the winning strategy in $G[\sigma_{-i}]$
will be the deviation that player $i$ uses to increase her payoff.
Let us explain how to construct the one-player game $G[\sigma_{-i}]$. 

\begin{cons}\label{cons:dev}
  We note $V$ the set of vertices in $G[\sigma_{-i}]$, $E$ the edge
  relation defined over $V\times \sfT^A \times V$, and the weight
  function $w$ is a mapping from $V\times \sfT^A \to \ZZ^n$.

  Let $\AA_{i} = (Q, \sfT^{A_{\phi_i}}, q_0, \Delta, F)$ be an
  automaton accepting the language $\LL_{\phi_i}$.

  The graph $G[\sigma_{-i}]$ is obtained as follows:
  \begin{itemize}
  \item The vertices are
    $V = Q \times \prod_{j\in N\setminus\set{i}}M_j$.
  \item Let $v$ be a vertex in $V$, for $j \in N\setminus\set{i}, v_j$
    refers to the $j$-th component of $v$ and $v_i$ is the projection
    over $Q$.  For $(u,v) \in V\times V$, and for every valuation
    $X \in \sfT^{A}$ we have $(u,X,v)$ in $E$ if
    \begin{enumerate}
    \item[$i)$] there exists $Y\in \sfT^{A_{\phi_i}}$ such that
      $(u_i, Y ,v_i) \in \Delta$ and $Y\subseteq X$,
    \item[$ii)$] the set
      $Z = Y\cup \bigcup_{j \in N\setminus\set{i}}\ch{\sigma}_j(u_j)
      \subseteq X$ and is consistent over $A_{\phi_i}$ i.e.
      \[
      \forall p \in A_{\phi_i},~ (p,\top) \in Z \implies (p,\bot)
      \not\in Z \enspace,
      \]
    \item[$iii)$] for each $j\in N\setminus\set{i}$ we have
      $\up{\sigma}_j(u_j, X) = v_j$.
    \end{enumerate}
  \item The weight function is given by $\cost(\ch{\sigma}_j(u_j))$
    for every dimension $j \in N\setminus\set{i}$ and by
    $\sum_{p\in A_i}c(p,X(p))$ for dimension $i$.
  \item Finally, a vertex $v\in V$ is accepting if $v_i \in F$.
  \end{itemize}
\end{cons}

The intuition behind this construction is as follows. If player $i$ can
deviate rationally, then necessarily the new profile satisfies
$\phi_i$. This is why we use automaton $\AA_{\phi_i}$ whose language is
exactly those words that satisfy $\phi_i$.  Also, since we consider
only unilateral deviations, the actions leading to the satisfaction of
$\phi_i$ have to be compatible with the choices of other players, that
is $\sigma_{-i}$. This is ensured by $ii)$. Item $iii)$ is a
synchronisation between the action of the other player and the
deviation of player $i$.

Thanks to the following lemma, we show that
Construction~\ref{cons:dev} meets the desired intuition.

\begin{restatable}{lemma}{lmConsTwo}
  \label{lm:DevStrat}
  Let $\sigma$ be a finite memory profile, and $i$ be a player such
  that $\out_i(\sigma) = 0$ then, $i$ has a rational deviation iff
  there exists a winning strategy in $G[\sigma_{-i}]$.
\end{restatable}

As a consequence we obtain the core property for the existence of our
$\pspace$ algorithm.

\begin{restatable}{proposition}{propConsTwo}
  \label{prop:onthefly}
  Let $\sigma$ be a finite memory profile, and $i$ be a player such
  that $\out_i(\sigma) = 0$.  We can check whether $i$ has a rational
  deviation in $\pspace$.
\end{restatable}

\subsection{Proof of Theorem~\ref{th:NEM-PSPACEc}}

We recall Theorem~\ref{th:NEM-PSPACEc}
\thmMain*

\begin{proof}
  If the profile is not feasible, return ``no''.  Otherwise, guess a
  possible deviator $i$ (among the players with null payoff) and check
  whether she has a winning strategy in $G[\sigma_{-i}]$. Return
  ``no'' iff she has a winning
  strategy. Lemma~\ref{prop:feasCor} and Lemma~\ref{lm:DevStrat}
  justify the correctness. Proposition~\ref{prop:feasCheck} and
  Proposition~\ref{prop:onthefly} justify the
  upper-bound complexity.



  To establish the hardness, one needs to notice that any BG is an EBG
  with endowment $\set{0}^N$ and
  $c : A\times\sfT \to \set{0}$.  Thus the $\pspace$ lower
  bound established in~\cite[Prop.~2]{GutierrezHW15} holds for EBGs
  with LTL specifications. Since the proof is a reduction from LTL
  satisfiability to one-player iterated boolean games, NEM is hard
  even when there is only one player.
\end{proof}



\section{Resource redistributions}
\label{sec:engineering}

Having characterised the complexity of the problem of deciding whether
a strategy profile of an iterated EBG is a Nash equilibrium, we will
see how we can easily tackle derived decision problems for engineering
Electric Boolean Games. 

A resource redistribution for an EBG $\BB = (N, \Sigma, \Phi, c, e )$
is an endowment function $e': N \to \NN$ such that
\[\sum_{i \in N} e(i) = \sum_{i \in N} e'(i).\]

\begin{remark}
  \label{prop:number-redistributions}
  Let an EBG $\BB = (N, \Sigma, \Phi, c, e )$. There is finite number
  of resource redistributions for $\BB$.
\end{remark}

In~\cite{HarrensteinTW15}, the authors studied the problems of
determining whether there is a resource redistribution such that a
strategy profile is a Nash Equilibrium (rational construction), and of
determining whether there is a resource redistribution such that a
strategy profile is not a Nash Equilibrium (rational elimination). For
the iterated setting we propose the following decision problems.

\begin{definition}[Construction and elimination]
  Let $\BB$ be an electric boolean game, and $\profile$ be a finite
  memory strategy profile. The Rational Construction (RC) problem asks
  whether there is a resource redistribution such that $\profile$ is a
  Nash equilibrium.The Rational Elimination (RE) problem asks whether
  there is a resource redistribution such that $\profile$ is not a
  Nash equilibrium.
\end{definition}



\begin{restatable}{theorem}{thmEng}\label{th:RC-RE-PSPACEc}
  The RC problem and the RE problem are $\pspace$-complete.
\end{restatable}

The non-deterministic procedures outlined in the proof of
Theorem~\ref{th:RC-RE-PSPACEc} are sufficient to characterise an
optimal upper-bound of the problems. In the case of RE, there exists a
more practical deterministic algorithm. Indeed, the result
of~\cite[Corr.~4]{HarrensteinTW15} carries over in the iterated
setting.

\begin{restatable}{proposition}{propEng}
\label{prop:eng}
  Let an endowment $e$ be given. The endowment $e^i$ is the resource
  redistribution of $e$ such that all resources are allocated to
  player~$i$. The strategy profile $\sigma$ is eliminable in
  $\BB^{c,e}$ iff for some player~$i$,
  $\sigma \not \in \NE(\BB^{c,e^i})$.
\end{restatable}

This hints at a ``more practical'' algorithm to solve RE: for each
player~$i$, test whether $\sigma \not\in \NE(\BB^{c,e^i})$. Return
``yes'' as soon as a test succeeds. Return ``no'' when all $|N|$ tests
failed.

\section{Conclusion}
In this paper we presented a preliminary result on the Electric
Boolean Games introduced in~\cite{HarrensteinTW15}. We considered the
iterated setting where the objectives are specified as LTL
formulas. We showed the $\pspace$-completness of Nash equilibrium
membership, thus matching the complexity bounds
of~\cite{GutierrezHW15} for the non quantitative setting of iterated
Boolean Games. In order to establish this result, we extended existing
techniques for plain LTL to an extension of LTL with electric
constraints.  This result is used to characterise the complexity of
two problems of resource redistribution that can serve at
social-welfare engineering.

As future research direction, we plan to investigate the Nash
equilibrium non-emptyness and Nash equilibrium synthesis. We believe that
Construction~\ref{cons:dev} can be extended in order to construct a
concurrent game with the property that it contains a pure Nash
equilibrium iff the electric boolean game does. To the best of our
knowledge, the obtained class of concurrent games is rather novel and
has yet to be studied.


\bibliographystyle{eptcs} 
\bibliography{ourbiblio.bib}



\end{document}